\documentclass[11pt]{article}
\usepackage[utf8]{inputenc}
\usepackage[margin=1in]{geometry}
\usepackage{amsmath, amsthm, amssymb}
\usepackage{mathtools}
\mathtoolsset{showonlyrefs}
\usepackage{dsfont}
\usepackage[usenames,dvipsnames]{xcolor}
\usepackage{thm-restate}
\definecolor{Gred}{RGB}{219, 50, 54}
\definecolor{ToCgreen}{RGB}{0, 128, 0}

\usepackage{authblk}

\usepackage{hyperref}
\hypersetup{
      colorlinks=true,
  citecolor=ToCgreen,
  linkcolor=Sepia,
  filecolor=Gred,
  urlcolor=Gred
  }

\title{Purest Quantum State Identification}
\author[1]{Yingqi Yu \thanks{Email:yingqiyu@mail.ustc.edu.cn}}
\author[1]{Honglin Chen \thanks{Email:chl777@mail.ustc.edu.cn}}
\author[1]{Jun Wu \thanks{Email:devilu@mail.ustc.edu.cn}}
\author[1]{Wei Xie \thanks{Email:xxieww@ustc.edu.cn \textit{Corresponding author}}}
\author[1,2]{Xiangyang Li \thanks{Email:xiangyangli@ustc.edu.cn \textit{Corresponding author}   }}

\affil[1]{School of Computer Science and Technology, University of Science and Technology of China}
\affil[2]{Hefei National Laboratory, University of Science and Technology of China, Hefei 230088, China}

\date{}

\usepackage{algorithmic}
\usepackage{algorithm}
\usepackage{booktabs} % for professional tables
\usepackage[braket]{qcircuit}

\theoremstyle{plain}
\newtheorem{theorem}{Theorem}[section]

\newtheorem{lemma}[theorem]{Lemma}
\newtheorem{problem}[theorem]{Problem}
\newtheorem{corollary}[theorem]{Corollary}
\theoremstyle{definition}
\newtheorem{definition}[theorem]{Definition}

\theoremstyle{remark}

\newenvironment{proof-sketch}{\begin{proof}[Proof Sketch]}{\end{proof}}

\usepackage{dsfont}
\newcommand{\indicator}[1]{\mathds{1}#1}

\newcommand{\bx}{{\boldsymbol x}}

\newcommand{\Tr}{\mathrm{Tr}}
\newcommand{\Var}{\mathrm{Var}}
\newcommand{\Haar}{\mathrm{Haar}}
\newcommand{\diag}{\mathrm{diag}}

\DeclarePairedDelimiter{\brk}{[}{]}

\makeatletter
\def\Pr{\@ifnextchar[{\@witha}{\@withouta}}
\def\@witha[#1]{\mathop{{}\operator@font Pr}_{#1}\brk}
\def\@withouta{\mathop{{}\operator@font Pr}\brk}
\makeatother

\makeatletter
\def\E{\@ifnextchar[{\@withb}{\@withoutb}}
\def\@withb[#1]{\mathop{{}\mathbb{E}}_{#1}\brk}
\def\@withoutb{\mathop{{}\mathbb{E}}\brk}
\makeatother

\usepackage{enumitem}

\makeatletter
\renewcommand{\paragraph}{%
  \@startsection{paragraph}{4}%
  {\z@}{1.25ex \@plus 1ex \@minus .2ex}{-1em}%
  {\normalfont\normalsize\bfseries}%
}
\makeatother

\usepackage{accents}

\DeclareMathAccent{\wtilde}{\mathord}{largesymbols}{"65}

\begin{document}

\clearpage
\pagestyle{plain}
\pagenumbering{arabic}

\maketitle

\begin{abstract}

Quantum noise constitutes a fundamental obstacle to realizing practical quantum technologies. To address the pivotal challenge of identifying quantum systems least affected by noise, we introduce the purest quantum state identification, which can be used to improve the accuracy of quantum computation and communication. We formulate a rigorous paradigm for identifying the purest quantum state among $K$ unknown $n$-qubit quantum states using total $N$ quantum state copies. For incoherent strategies, we derive the first adaptive algorithm achieving error probability $\exp\left(- \Omega\left(\frac{N H_1}{\log(K) 2^n }\right) \right)$, fundamentally improving quantum property learning through measurement optimization. By developing a coherent measurement protocol with error bound $\exp\left(- \Omega\left(\frac{N H_2}{\log(K) }\right) \right)$, we demonstrate a significant separation from incoherent strategies, formally quantifying the power of quantum memory and coherent measurement. Furthermore, we establish a lower bound by demonstrating that all strategies with fixed two-outcome incoherent POVM must suffer error probability exceeding $ \exp\left( - O\left(\frac{NH_1}{2^n}\right)\right)$. This research advances the characterization of quantum noise through efficient learning frameworks. Our results establish theoretical foundations for noise-adaptive quantum property learning while delivering practical protocols for enhancing the reliability of quantum hardware.

\end{abstract}

\section{Introduction}

Quantum computers possess the potential to solve specific problems with significantly greater efficiency than classical computers \cite{shor1994algorithms,daley2022practical}. However, as quantum computing currently resides in the Noisy Intermediate-Scale Quantum (NISQ) era, the number of quantum bits (qubits) in quantum devices is constrained \cite{lau2022nisq,preskill2018quantum}. Furthermore, these qubits are susceptible to noise during operations, resulting in a lack of complete control \cite{google2023suppressing}. Numerous algorithms have been developed to enhance the precision of quantum operations to maximize the utility of existing quantum devices \cite{ramezani2020machine}. Nonetheless, owing to variations in device implementation and environmental factors, the effectiveness of these algorithms differs markedly among various quantum devices \cite{gyongyosi2019survey}. Thus, identifying the optimal quantum device, quantum algorithm, or quantum channel presents a critical issue. To achieve this, we must measure quantum states in multiple quantum systems, analyze their properties, and select the best one with the most significant probability.

However, quantum learning is challenging \cite{wright2016learn, montanaro2013survey, aaronson2018online, chen2022tight, chen2022toward, bubeck2020entanglement, fawzi2023quantum}. The properties of superposition and entanglement in quantum systems result in exponential growth of the state space as the qubit number $ n $ increases \cite{gyongyosi2019survey}. Consequently, the complexity of measuring the state of a quantum system escalates rapidly with the qubit number. When aiming to capture all the information of an $ n $-qubit quantum system, the learner takes $ \Theta(2^{3n})$ samples and measurements \cite{chen2023does}. Therefore, identifying the best quantum system through tomography will incur substantial costs. Determining how to distribute limited measurements across multiple unknown quantum systems — and how to choose the right measurement basis — remains a key challenge in identifying the most suitable quantum system for a given task.

This paper investigates the problem of the purest quantum state identification (PQSI). Pure quantum states are critically important in quantum computing and quantum communication, serving as a fundamental requirement for various quantum algorithms. Furthermore, identifying the quantum state least affected by noise is a significant issue , especially in the NISQ era. The problem PQSI has a wide range of application scenarios, including quantum state preparation \cite{plesch2011quantum}, selection of quantum channels \cite{chen2023unitarity}, and initialization of quantum algorithms \cite{shor1994algorithms}. Addressing this problem will substantially contribute to the advancement of quantum technology. %

This paper makes the following key contributions:
    
\textbf{Problem Model.} To the best of our knowledge, this is the first study dedicated to the best quantum state identification. In this paper, we focus on the issue of the purest quantum state identification, i.e., identifying the purest quantum state from a set of available quantum states. In this problem, while allocating copies to different quantum systems, the learner also needs to select the basis for quantum measurement, which significantly increases its complexity. Due to the limited number of qubits in existing quantum devices, these systems may not be able to facilitate large-scale quantum measurement \cite{yu2021sample, heshami2016quantum}. Consequently, while coherent (multi-copy) measurement techniques can more efficiently acquire information about quantum states, research on incoherent (single-copy) measurement is more practical. We formalize the problem of purest quantum state identification with incoherent (single-copy) measurement as follows:

\begin{problem}[Purest quantum state identification (PQSI) with incoherent (single-copy) measurement]\label{problem:PQSI_restatement} 
    There is a set of $K$ unknown $n$-qubit quantum states represented as $ S = \{ \rho_1, \ldots, \rho_K \} $. The learner aims to identify the purest quantum state in $S$ using a total of $N$ copies across all states. The problem protocol at round $t \in \{1, ..., N\}$ is as follows:
    
    \ \ \textbullet \ The learner chooses a quantum state $\sigma_t$ from the set $S$ and gets a copy of it.
    
    \ \ \textbullet \ The learner selects a POVM and uses it to measure $\sigma_t$, after which $\sigma_t$ is destroyed.
    
    Upon completing $ N $ measurements, the learner selects a quantum state $ \rho' \in S$ based on the measurement outcomes. The objective of the learner is to maximize the probability of identifying the purest quantum state, i.e., $ \mathbb{P}(\rho' \in \arg\max_{\rho \in S} \Tr(\rho^2)) $.

\end{problem}

In general, when coherent measurements are available, we formalize the problem as follows:
    
\begin{problem}[Purest quantum state identification(PQSI) with coherent (multi-copy) measurement]\label{problem:PQSI}
    There is a set of $K$ unknown $n$-qubit quantum states represented as $ S = \{ \rho_1, \ldots, \rho_K \} $. The learner aims to identify the purest quantum state in $S$ using a total of $N$ copies across all states. Let $ N_t $ denote the number of quantum state copies that remain available in round $ t $ and $N_1 = N$. When $N_t > 0 $, the problem protocol at round $t\in \mathbb{N}^+$ is as follows:
    
    \ \ \textbullet \ The learner decides the number of copies of the quantum state $ \rho_i $ to acquire, denoted as $ s(i,t) $, where $ s(i,t) \geq 0 $ and $ 1 \leq \sum_{i=1}^K s(i,t) \leq N_t $. The quantum state constructed from these copies can be represented as $ \sigma_t =  \rho_1^{\otimes s(1,t)} \otimes \ldots \otimes \rho_K^{\otimes s(K,t)}$. Let $ N_{t+1} = N_t - \sum_{i=1}^K s(i,t) $.
    
    \ \ \textbullet \ The learner selects an entangled POVM and uses it to measure $\sigma_t$,  after which $\sigma_t$ is destroyed.
    
    Upon completing measurements of all $N$ copies of quantum states, the learner selects a quantum state $ \rho' \in S$ based on the measurement outcomes. The objective of the learner is to maximize the probability of identifying the purest quantum state $ \mathbb{P}(\rho' \in \arg\max_{\rho \in S} \Tr(\rho^2)) $.
\end{problem}

\textbf{Algorithms.} We develop two distinct algorithms to address this problem in different settings. To simplify the expression, for $i \in \{1,...,K\}$, let $\rho_{(i)}$ be the $i$-th purest quantum state in $S$ and
%\begin{equation*}
    $\Delta_{(i)} =\Tr\left(\rho^2_{(1)}\right) - \Tr\left(\rho^2_{(i)}\right).$
%\end{equation*}
In scenarios where only single-copy measurements are available, we developed the incoherent measurement based algorithm IM-PQSI. We employ Haar unitary matrices to generate random measurement basis and to measure quantum states in $ S $. This measurement approach allows us to evaluate the purity of quantum states by analyzing measurement expectations and variances. The Haar unitary matrices utilized in IM-PQSI belong to the class of unitary 4-designs, which can be efficiently simulated on a quantum computer up to inverse-exponential trace distance \cite{ma2024construct}. During the evaluation of quantum state purity, our algorithms dynamically allocate copies to various quantum states. By assigning more copies to states with higher purity, we enhance the estimation accuracy of these states, thereby increasing the likelihood of identifying the purest quantum state. Furthermore, we improve the performance of the algorithm for identifying the optimal quantum state by balancing the selection of measurement basis with the number of measurements conducted for each base. The error probability of this algorithm satisfies the following theorem:

\begin{theorem}[Informal version of Theorem \ref{thm:alg_IMPQSI}] 
    There exists an algorithm that solves the problem of the purest quantum state identification with incoherent measurement whose error probability satisfies 
    \begin{equation}
        e_N \leq \exp\left(- \Omega\left(\frac{N H_1}{\log(K) 2^n }\right) \right),
    \end{equation}
    where $H_1 = \min_{i \in \{2,...,K\}} \frac{\Delta_{(i)}}{i}$.
\end{theorem}
    
Conversely, when multi-copy measurements are accessible, we developed the coherent measurement based algorithm CM-PQSI. We use the SWAP test to estimate the purity of quantum states in this algorithm, and its error probability satisfies the following theorem:
    
\begin{theorem}[Informal version of Theorem \ref{thm:alg_CMPQSI_coherent}] 
    There exists an algorithm that solves the problem of the purest quantum state identification with coherent measurement whose error probability satisfies 
    \begin{equation}
        e_N \leq \exp\left(- \Omega\left(\frac{N H_2}{\log(K) }\right) \right),
    \end{equation}
    where $H_2 = \min_{i \in \{2,...,K\}} \frac{\Delta^2_{(i)}}{i}$.
\end{theorem}

By comparing the error probability upper bound of these two algorithms, we can identify the advantages of quantum memory.

\textbullet \  \textbf{Lower Bound.} Analyzing the complexity of problems related to quantum testing is usually challenging. These problems often use Haar unitary matrices to construct specific cases that are difficult to distinguish, reflecting the complexity of these issues \cite{anshu2022distributed, gong2024sample}. However, the representation-theoretic structure of Haar unitary matrices is complicated, which makes them difficult to analyze. When distinguishing quantum states from two alternative sets, previous work usually assumes that one of them is the maximally mixed state or pure state to reduce the difficulty of analysis. In our problem, the learner must frequently distinguish between quantum states with different purity, making this analysis method unsuitable. 

For incoherent measurements, we utilize the properties of Haar unitary matrices to formulate a related problem named Purest Random Quantum State Identification (PRQSI). In this context, the learner is required to consider a specific quantum state distribution constructed from Haar unitary matrices. We demonstrate that the lower bound of the error probability for the PRQSI problem is also applicable as the lower bound for the error probability of the PQSI problem. Furthermore, we show that the measurement outcomes generated by three-fourths of the Haar unitary matrices are insufficient for the learner to easily differentiate among them, presenting a challenge in identification. Additionally, we demonstrate that in any PRQSI problem employing a fixed two-outcome Positive Operator-Valued Measure, the measurement outcomes exhibit a Bernoulli distribution that is inherently difficult to distinguish. Consequently, we derive the lower bounds for the error probability of the PRQSI, which are also the lower bounds for the error probability of the PQSI problem, as follows:

%will have an error probability of at least $ \exp\left( \Omega\left(- \frac{NH_1}{2^n}\right)\right)$. 
\begin{theorem}[Informal version of Theorem \ref{thm:PQSI lower bound}] 
    For any algorithm $\mathcal{A}$ to solve the purest quantum state identification using fixed 2-outcome randomly incoherent POVM, there exists a set of quantum states which makes the error probability of $\mathcal{A}$ satisfies
    \begin{equation}
        e_N \geq \exp\left( -O\left(\frac{NH_1}{2^n}\right)\right),
    \end{equation}
    where $H_1 = \min_{i \in \{2,...,K\}} \frac{\Delta_{(i)}}{i}$.
\end{theorem}

\paragraph{Structure of the paper.} In Section \ref{sec: related_work}, we review relevant literature and discuss previous work related to our research. In Section \ref{sec: pre and model}, we provide preliminaries and notations used throughout this paper. In Section \ref{sec: PQSI_inco}, we give the algorithm to solve the purest quantum state identification with incoherent measurement. Section \ref{sec: lower_bound} gives the error probability lower bound for any algorithm $\mathcal{A}$ to solve the PQSI using two-outcome randomly incoherent POVM. We give the algorithm to solve the PQSI with coherent measurement in Section \ref{sec: PQSI_co}. Finally, we summarize the paper's content and present some related open problems in Section \ref{sec: conclusion}.

\section{Related Work}
\label{sec: related_work}

The problem of Purest Quantum State Identification (PQSI) can be viewed as learning the properties of a set of quantum states.

\paragraph{Quantum learning and testing.}  Quantum learning and testing \cite{montanaro2013survey, aharonov2022quantum} is a vital area of research in quantum computing and quantum communication. There are extensive investigations conducted to understand the complexities of various measurements. 

Quantum state tomography \cite{banaszek2013focus, gross2010quantum, chen2023does, chen2022exponential} involves obtaining complete information about the density matrix of a quantum state through measurements. While this technique can be employed to tackle the PQSI problem, it incurs significant sampling costs. For the quantum state certification \cite{bubeck2020entanglement,chen2022tight,wright2016learn,chen2022toward}, the target is to determine whether a quantum state is close to a specific target quantum state. Our problem can be viewed as identifying the quantum state that is the farthest from the maximally mixed state. However, this problem is focused on quantum testing and does not deal with distance estimation. Therefore, these methods cannot be applied to the PQSI problem. Another category of problems relates to inner product estimation between two quantum states \cite{anshu2022distributed,hinsche2024efficient, zhu2022cross,huang2020predicting}. When proving lower bounds, this category often significantly restricts the quantum states for distinction. 

The relevant literature employs two general approaches to establish problem complexity. The first approach involves constructing counterexamples using Haar unitary matrices \cite{anshu2022distributed,chen2022exponential,chen2022toward, bubeck2020entanglement}. However, the representation-theoretic structure of Haar unitary matrices is intricate \cite{mele2024introduction}, which makes it difficult to use. The other approach uses Gaussian Orthogonal Ensemble (GOE) matrices to create counterexamples \cite{chen2022tight,chen2023does}. However, when using GOE to prove the lower bound, the distance between quantum states is in a specific range rather than a fixed number, making it unsuitable for the PQSI problem. 

% To address the problem of PQSI, we designed and designed a new scheme to compare quantum states of varying purity to.

\paragraph{Classical Best arm identification.} To the best of our knowledge, our work is the first to consider the best quantum state identification. Among the classical learning tasks, the best arm identification \cite{audibert2010best, garivier2016optimal, russo2016simple, 6814096, gabillon2012best} has been extensively studied, and is divided into two categories: fixed budget \cite{bechhofer1968single} and fixed confidence \cite{paulson1964sequential}. However, the existing research can only deal with the problem under specific distributions. % For the case of a fixed budget, \cite{audibert2010best} deals with the problem for Bernoulli bandit models. They proved an asymptotic lower bound under this assumption. For the case of fixed confidence, most works \cite{huang2017structured, jourdan2023varepsilon, degenne2019pure} only deal with the bandit models with a one-parameter exponential family. 
These limitations restrict the algorithm's applicability and leave considerable room for further research on this issue. In our problem, we must select an appropriate POVM basis while choosing the quantum state in each round. The quantum state space and the POVM space grows exponentially with the increase in qubits, which makes this problem significantly more challenging than solving a classical problem of the best arm identification.
\section{Preliminaries and Notations}
\label{sec: pre and model}

\begin{table}[t]
\caption{Description of commonly used-notations}
\label{table:notations}
\vskip 0.15in
\begin{center}
\begin{small}
\begin{tabular}{c|l}
\hline
Notation & Description \\
\hline
$n$  & the qubit number of the quantum state\\
\hline
$d$  & $d =2^n$ \\
\hline
$S$ & the set of the unknown quantum state \\
\hline
$K$ & $K = |S|$\\
%\hline
%$[K]$ & $[K] = \{1,\ldots,K\}$\\
\hline
$\rho_i$ & the $i$-th quantum state in $S$ \\
\hline
$\rho_{(i)}$ & the $i$-th purest quantum state in $S$ \\
\hline
$\rho^\star = \rho_{i^\star}$ & the purest quantum state in $S$ \\
\hline
$\Delta_i$ & $\Delta_i =\Tr(\rho_{i^\star}^2) - \Tr(\rho^2_{i})$\\
\hline
$\Delta_{(i)}$ & $\Delta_{(i)} =\Tr(\rho_{i^\star}^2) - \Tr(\rho^2_{(i)})$ \\
\hline
$\{|i \rangle\langle i|\}^{d-1}_{i=0}$ & a fixed orthogonal basis in $\mathbb{C}^{d \times d}$\\
\hline
$I_d$ & $d$-dimensional identity matrix \\
\hline
$\mathbb{U}(d)$ & the set of $d \times d$ unitary matrix\\
\hline
\end{tabular}
\end{small}
\end{center}
\vskip -0.1in
\end{table}

In this work, we will use Dirac's bra-ket notation, where $| v \rangle \in \mathbb{C}^d$ denotes a column vector and $\langle v| = | v \rangle^\dagger$. 
Specifically, for all $ i \in \{0,...,2^n-1\}$, let $| i \rangle$ denote a column vector whose $(i+1)$-th element is 1, and all other elements are 0. 
%\begin{equation*}
%    | i \rangle = (\ \underbrace{0,...,0}_{i \ \text{times} 0}\ ,\ 1,\underbrace{0,...,0}_{(n-1-i) \times 0})^T.    
%\end{equation*}
Then $\{|i\rangle \langle i | \}_{i=0}^{d-1}$ is a fixed orthogonal basis in $\mathbb{C}^{d \times d}$.

\paragraph{Quantum State.} Let $ d = 2^n $ denote the dimension of a $n$-qubit quantum system. An $ n $-qubit quantum state can be represented by a density matrix $ \rho \in \mathbb{C}^{d \times d} $, which is Hermitian and trace-1 positive semi-definite. In particular, an $n$-qubit pure quantum state can be represented by a unit vector $|\psi \rangle \in \mathbb{C}^d$. The purity of a quantum state $\rho$ is $\Tr(\rho^2)$.

% \textbf{Quantum measurement.} 量子测量，使用 $E$ 对量子态 $\rho$ 进行测量，会以 $\Tr(E\rho)$ 的概率输出 1，以 $1 - \Tr(E\rho)$ 输出为 0。

\paragraph{Quantum Measurement.} Quantum measurements are usually described by a Positive Operator-Valued Measure (POVM), which produces probabilistic outcomes. The formal definition of a POVM is as follows:

\begin{definition}[Positive Operator-valued measurement (POVM), see e.g. \cite{nielsen2010quantum}]
    An $n$-qubit positive operator-valued measurement $\mathcal{M}$ can be represented as a collection of positive semi-definite matrices $\mathcal{M} =\{M_m\}_m$, where $M_m \in \mathbb{C}^{d\times d}$ and $\sum_m M_m = I_d$. When using $\mathcal{M}$ to measure a quantum state $\rho$, the probability of outcome $m$ is $\Tr(M_m\rho)$, and the quantum state $\rho$ is destroyed.
\end{definition}

When the coherent measurement method is employed, the learner can perform entanglement measurements on quantum states $ \rho_1 \otimes \ldots \otimes \rho_m $. However, this approach necessitates the support of large-scale quantum devices and quantum memory, which are not feasible with current quantum technologies. Therefore, researching incoherent measurement methods applicable to NISQ-era quantum devices is of great significance. 

% \paragraph{Incoherent Measurements.} Intuitively, such an algorithm operates as follows: in each round $t \in \{1,...,N\}$, the learner chooses a quantum state in $S_\rho$ and measures its copy using a POVM $\mathcal{M}_t$, which could depend on the results of previous measurements. Then, the learner records the outcome and repeats this process in the following rounds. After having performed all $N$ samples and measurements, it must output a decision based on the sequence of outcomes it has received.

In this study, we aim to identify the purest quantum state from a set of unknown quantum states, achieving the highest probability through $ N $ measurements. For the purpose of simplicity, we will assume that there exists a unique optimal quantum state that is the purest in the set $ S $, denoted as $ \mu^\star = \mu_{i^\star} $. For $ i \neq i^\star $, we represent the purity difference between each non-optimal quantum state and the optimal quantum state using the following expression:
\begin{equation*}
    \Delta_i = \Tr\left(\rho_{i^\star}^{2}\right) - \Tr\left(\rho_{i}^{2}\right).
\end{equation*}

For $i \in \{1,...,K\}$, let $\rho_{(i)}$ be the $i$-th purest quantum state in $S$, then we have
\begin{equation*}
    \Tr\left(\rho_{i^\star}^{2}\right) = \Tr\left(\rho_{(1)}^{2}\right) > \Tr\left(\rho_{(2)}^{2}\right) \geq...\geq \Tr\left(\rho_{(K)}^{2}\right),
\end{equation*}
and
\begin{equation*}
    \Delta_{(2)} \leq \Delta_{(3)} \leq ... \leq \Delta_{(K)}.
\end{equation*}

Let $e_N$ denote the probability that the learner does not choose the purest quantum state in $S$ after $N$ samples and measurements, i.e., $e_N = \mathbb{P}\left(\rho' \notin \arg\max_\rho \Tr(\rho^2)\right)$. The learner's objective is to $\min e_N$.

We summarize key notations used throughout this paper in Table \ref{table:notations}.
\section{Algorithm for PQSI with incoherent measurement}
\label{sec: PQSI_inco}

The current era of quantum computing presents limitations in the number of available qubits. This may hinder the measurement of multiple quantum states operated jointly. In this section, we use incoherent (single-copy) measurement methods in each round to select the purest quantum state from the set of quantum states.

The algorithm we designed to solve the PQSI problem with incoherent measurement is shown in Algorithm \ref{alg:PQSI_incoherent}. In our algorithm, we utilize Haar unitary matrices to construct multiple random measurement bases for measuring quantum state copies. For the measurement results obtained from each constructed measurement base, we demonstrate that measuring a quantum state $ m $ times with this base allows for estimating the quantum state's purity based on the probability of identical outcomes. Formally, for a $ d $-outcome POVM, let the result of the $ i $-th measurement be $ x_i \in \{0, \ldots, d-1\} $. We can estimate the purity of the quantum state using the probability of identical outcomes $ \tilde{g} = \frac{1}{m(m-1)} \sum_{i=1}^{m} \sum_{j \neq i} \indicator\{x_i = x_j\} $. By calculating the expectation and variance of $\tilde{g}$, we can demonstrate the probability of identifying the purest quantum state by our algorithm.

Algorithm \ref{alg:PQSI_incoherent} includes two random processes. The first process involves the random selection of Haar unitary matrices to construct measurement bases, while the second process entails the random acquisition of measurement results using these bases. Balancing the estimation errors introduced by these two processes is crucial for enhancing the algorithm's accuracy. When the purity difference between quantum states is substantial, each measurement base requires fewer instances to distinguish purity differences among states effectively. Conversely, when the purity difference is minimal, more measurements are necessary under a single measurement base to gather sufficient information, limiting the number of measurement bases that can be utilized. Furthermore, when the purity difference is extremely slight, the choice of measurement basis significantly influences the measurement outcomes. Performing $ \Theta(d^2) $ measurements under a single measurement base is essential for adequately differentiating the outcomes between quantum states. Additionally, increasing the number of measurements under one measurement base will not improve the accuracy of purity estimation. Given these considerations, we balance the number of measurement bases and the number of measurements per base to enhance the probability of identifying the optimal quantum state.

In Algorithm 1, we use Haar unitary matrices to construct random measurement bases. Haar unitary matrices are extensively employed in quantum property testing and learning theory \cite{brandao2021models, anshu2022distributed, huang2020predicting}. Achieving exact Haar randomness is challenging; therefore, unitary $t$-designs are typically used to approximate Haar unitary matrices. Haar unitary matrices used in Algorithm \ref{alg:PQSI_incoherent} can be approximated using unitary $4$-designs. As established in \cite{ma2024construct}, unitary $t$-designs can be efficiently implemented on quantum computers with exponentially minor approximation errors. This theoretical foundation ensures the practical feasibility and mathematical validity of employing Haar unitary matrices in our algorithmic framework.

% random samples to estimate the expected value of a random process and estimate the purity of quantum states. For a distribution $p$ supported on $\{0,...,d-1\}$, let $p_i$ denote the probability of observing $i$ in the distribution $p$. We employ the following purity estimation method to estimate the $\sum_{i=0}^{d-1}p_i^2$.

%\begin{definition}[purity estimator]
%    \label{def: purity_estimator}
%    Given $m$ samples $x_1,...,x_m \sim p$, the purity estimator is defined as 
%    \begin{equation} \label{eq:purity_estimator}
%        \tilde{g} = \frac{1}{m^2} \sum_{i=0}^{d-1} \left[ \sum_{j=1}^m \indicator\{x_j =i\}\right]^2 - \frac{1}{m}.
%    \end{equation}
%\end{definition}

%We can prove that the expectation of the purity estimation $\tilde{g}$ is $\frac{m-1}{m}\sum_{i=0}^{d-1} p_i^2$.

% By leveraging the properties of Haar unitary matrices, we can connect the purity of a quantum state to the estimation of $ \sum_{i=0}^{d-1} p_i^2 $ for a classical distribution $ p $. We use the purity estimator defined in Definition $\ref{def: purity_estimator}$ to perform this estimation. Additionally, we utilize the successive reject algorithm to distribute a limited number of sampling times across the available quantum states, thereby enhancing the algorithm's accuracy. 

\begin{algorithm}[htb]
    \caption{Incoherent measurement based algorithm for solving PQSI problem (IM-PQSI)}
    \label{alg:PQSI_incoherent}
    \begin{algorithmic}
    \STATE {\bfseries Input:} Copy access to $S = \{\rho_1,...,\rho_K\}$, sample number $N$.
    \STATE {\bfseries Initialization:} Set $S_0 = \{\rho_1,...,\rho_K\}$, $\overline{\log}(K) = \frac{1}{2} + \sum_{i=2}^K \frac{1}{i}$, $N_0=0$ and $N_k = \left\lceil \frac{1}{\overline{\log}(K)} \frac{N-K}{K+1-k}\right\rceil$, for $k \in \{1,...,K-1\}$. Sample $\lfloor {N}/m \rfloor$ random unitary matrix $U_1,...,U_{\lfloor N/m \rfloor}$ according to the Haar measure.
    \FOR{k=1,..., K-1}
        \FOR{$\rho \in S_{k-1}$ and $j \in \{\lfloor \frac{N_{k-1}}{m}\rfloor +1,..., \lfloor \frac{N_{k}}{m}\rfloor \}$}
        \STATE Measure $m$ copies of $\rho$ in the basis $\{U_j^\dagger|i\rangle \langle i | U_j\}_{i=0}^{d-1}$ and set the outputs as $x{(\rho, j,1)},..., x{(\rho, j,m)}$.
        %\STATE Compute the purity estimator \eqref{eq:purity_estimator} using $x{(\rho, j,1)}$, $...$, $x{(\rho, j,m)}$ denoted as $\tilde{g}{(\rho,j)}$.
        \STATE Let $\tilde{g}(\rho,j) = \frac{1}{m^2} \sum_{i=0}^{d-1} \left[ \sum_{l=1}^m \indicator\{x(\rho,j,l) =i\}\right]^2 - \frac{1}{m}$.
        \ENDFOR
        \STATE  Let $w(\rho,k) = \frac{1}{\lfloor \frac{N_{k}}{m}\rfloor} \sum_{j=1}^{\lfloor \frac{N_{k}}{m}\rfloor} \tilde{g}{(\rho,j)}$.
        \STATE Let $S_k = S_{k-1} \setminus \arg \min_{\rho \in S_{k-1}} w(\rho,k) $.
    \ENDFOR
    \STATE Output the quantum state $\rho'$ in $S_{k-1}$.

\end{algorithmic}
\end{algorithm}

% \begin{proof}
%     \begin{equation}
%         \begin{aligned}
%             \mathbb{E}[w(\rho,k)] = & \mathbb{E}[\tilde{g}_{\rho,j}] \\
%             = & \frac{m-1}{m}\sum_{i=0}^{d-1} \mathbb{E}_{\psi \sim\mathbb{C}^d} (\langle \psi | \rho | \psi \rangle)^2 \\
%             = & \frac{(m-1)(1+ \Tr(\rho^2))}{m(d+1)}.
%         \end{aligned}
%     \end{equation}
% \end{proof}

Then, we give the error probability upper bound of IM-PQSI in the following theorem.

\begin{theorem}
    \label{thm:alg_IMPQSI}
    For $ i \in \{1,...,K\}, \Delta_i \geq c > \frac{1}{d^2}$, where $c$ is a constant. Set $m=\lceil \frac{1}{\sqrt{c}} \rceil$ in Algorithm \ref{alg:PQSI_incoherent}. The error probability of Algorithm \ref{alg:PQSI_incoherent} satisfies
    \begin{equation*}
         e_N \leq \frac{K(K-1)}{2}\exp\left(-\Omega \left( \frac{\sqrt{c}N H_1}{\overline{\log}(K) d }\right) \right),
    \end{equation*}
    where $H_1 = \min_{i \in \{2,...,K\}} \frac{\Delta_{(i)}}{i}$.
\end{theorem}

\begin{proof-sketch}
    The expectation of $w(\rho,k)$ and $\tilde{g}_{\rho,j}$ in Algorithm \ref{alg:PQSI_incoherent} satisfies 
    \begin{equation}
        \mathbb{E}[w(\rho,k)] = \mathbb{E}[\tilde{g}_{\rho,j}] = \frac{(m-1)(1+ \Tr(\rho^2))}{m(d+1)},
    \end{equation}
    and the variance of $\tilde{g}(\rho,j)$ satisfies
    \begin{equation}
        \Var(\tilde{g}(\rho,j)) = O\left(\frac{1}{d^3} + \frac{1}{m^2d} + \frac{1}{md^2}\right).
    \end{equation}
    By Bernstein's inequality and the union bound of error probability, we have 
    \begin{equation}
        e_n  \leq \frac{K(K-1)}{2} \exp\left( -\Omega\left(\frac{\sqrt{c}NH_1}{\overline{\log}(K)d}\right)\right),
    \end{equation} where $H_1 = \min_{i \in \{1,...,K\}} \frac{\Delta_{(i)}}{i}$. The proof details are provided in Appendix \ref{subsec:PQSI_alg1}.
\end{proof-sketch}

The dimension $d=2^n$ increases exponentially with the number of qubits $n$. As $n$ increases, $\frac{1}{d^2}$ tends to 0. Therefore, in Theorem \ref{thm:alg_IMPQSI}, we assume that for any quantum state $\rho \in S_\rho$, $\Tr(\rho^{\star2}) - \Tr(\rho^{2}) \geq c > \frac{1}{d^2}$. If we can not make this assumption, we can derive the following conclusion:

\begin{lemma}
    \label{lem2:PQSI_alg}
    Set $m=d$ in Algorithm \ref{alg:PQSI_incoherent}. The probability of error of Algorithm \ref{alg:PQSI_incoherent} satisfies
    \begin{equation*}
    \begin{aligned}
        e_N \leq \frac{K(K-1)}{2} \exp\left(- \Omega\left(\min\left(\frac{N H_2}{\overline{\log}(K)}, \frac{N H_1}{\overline{\log}(K)d^2}\right)\right) \right),
    \end{aligned}
    \end{equation*}
    where $H_1 = \min_{i \in \{2,...,K\}} \frac{\Delta_{(i)}}{i}$, and $H_2 = \min_{i \in \{2,...,K\}} \frac{\Delta^2_{(i)}}{i}$.
\end{lemma}

Appendix \ref{subsec:PQSI_alg2} provides the proof details of Lemma \ref{lem2:PQSI_alg}.

In the field of quantum learning and testing, research on quantum channels constitutes a critical aspect. Evaluating the impact of noise on quantum channels can significantly enhance the accuracy of quantum computing and quantum communication \cite{chen2023unitarity}. Introduced a method for assessing the ``unitarity" of a quantum channel by evaluating the purity of a quantum state. Subsequently, we can utilize the algorithm IM-PQSI to identify the most ``unitary" quantum channel from a quantum channel set. Let $u_{(i)}$ denote the unitarity of the $i$-th most unitary quantum channel. We have the following corollary:

\begin{corollary}
    There exists an algorithm that solves the problem of the most ``unitary" channel identification with incoherent access whose error probability satisfies:
 \begin{equation}
        e_N \leq \exp\left(- \Omega\left(\frac{N H_u}{\log(K) 2^n }\right) \right),
    \end{equation}
    where $H_u = \min_{i \in \{2,...,K\}} \frac{u_{(1)} - u_{(i)}}{i}$.
\end{corollary}
\section{Lower bound for PQSI with incoherent measurement}
\label{sec: lower_bound}

In this section, we investigate the lower bound on the error probability for solving the problem of purest quantum state identification. This problem requires distinguishing between quantum states with different purity through sampling and measurement. Recent studies \cite{mele2024introduction} indicated that when a quantum state $\rho$ is rotated by a Haar unitary matrix and measured $N$ times, the output distribution can be calculated only if $\rho$ is either a pure state or a maximally mixed state. Consequently, the complexity analysis of testing problems often assumes that one of the quantum states is either a pure state or a maximally mixed state. This limitation presents significant challenges to our analysis.

To solve this problem, we reduce the problem of identifying the purest quantum state to the problem of identifying the purest random quantum state. This reduction allows us to retain the problem's complexity while enabling us to analyze the complexity by considering only a single problem instance. 

Next, we demonstrate that for any POVM base $\mathcal{M}$, there is a set of unitary matrices $\mathbb{U}(\mathcal{M})$ satisfying that (1) $\mathbb{P}_{U \sim \Haar} (U \in \mathbb{U}(\mathcal{M})) = \Omega(1)$; and (2) when the quantum states rotated by these unitary matrices, they are difficult to distinguish by the POVM base  $\mathcal{M}$.

At last, we only consider all the possible POVM $\mathcal{M}$ and their corresponding set of unitary matrix $\mathbb{U}(M)$. By analyzing the sampling distribution for specific POVM $\mathcal{M}$ and unitary matrix in $\mathbb{U}(M)$, we reduce the problem into a classical problem for resolution and provide a lower bound for the purest quantum state identification.

Similar to Definition 7 in \cite{gong2024sample}, we analyze the lower bound of the error probability for any algorithm solving the purest quantum state identification problem using a 2-outcome randomly incoherent POVM to evaluate the task's difficulty.

\begin{definition}[Randomly fixed incoherent two-outcome POVM] We say an algorithm $\mathcal{A}$ with a randomly fixed incoherent two-outcome POVM, if it proceeds as the following:  The algorithm $\mathcal{A}$ samples a POVM $\mathcal{M} = \{M_0, M_1 = I_d -M_0\}$ from a well-designed distribution of POVMs $\mathcal{D}_{\mathcal{M}}$ and performs the two-outcome single-copy POVM $\mathcal{M}$ on the copies of the quantum states.
\end{definition}

\subsection{Problem reduction}

In this subsection, we aim to demonstrate that if there exists a set of random quantum states $ T(x) = \{\tau_1(x), \ldots, \tau_K(x)\} $ which is difficult to identify the purest one in $T$, there also exists a corresponding set of quantum states $ S = \{\rho_1, \ldots, \rho_K\} $ where is difficult to identify the purest one in $S$. In this way, we only need to construct $ K $  random quantum states, which are hard to distinguish. Then, we can demonstrate the difficulty of the purest quantum state identification problem.

% To enhance our discussion, we will redefine the problem of the purest quantum state identification with incoherent measurement as follows:

To enhance our discussion, we define the problem of the purest random quantum state identification(PRQSI) with incoherent measurement as follows:

\begin{problem}[Purest random quantum state identification(PRQSI) with incoherent measurement] \label{problem:PRQSI}  Consider a set of $K$ unknown random quantum states, denote as $ T(U) = \{ \tau_1(U), \ldots, \tau_K(U) \} $, where $U$ samples from a fixed distribution $\mathcal{D}$. For each $ k \in [K] $ and $U \sim \mathcal{D}$, $\Tr ((\tau_k(U))^2) = z_k$. In each round $t \in\{1,...,N\}$, the learner selects an index $k_t \in [K]$ and a POVM $\mathcal{M}_t$. The learner obtains a copy of $\tau_{k_t}(U)$ and uses $\mathcal{M}_t$ to measure it. Upon completing $ N $ measurements, the learner selects an index $ k' \in [K] $ as the output. The objective of the learner is to maximize $ \mathbb{P}(k' \in \arg\max_{k \in [K]} z_k) $.
\end{problem} 

As shown in the following lemma, we can reduce the proof of the error probability lower bound for the PQSI problem into the proof of the lower bound for a specific instance of the PRQSI problem.

\begin{lemma}%[Le Cam's two-point method for the purest quantum state identification, see e.g. Lemma 1 in \cite{Yu1997} and Lemma 5.5 in \cite{9996689}]
    \label{lem:problem transformation}
    If there exists a set of random quantum states $T(U)$ in Problem \ref{problem:PRQSI} such that any algorithm $\mathcal{A}_T$ addressing Problem \ref{problem:PRQSI} cannot identify the purest random quantum state with an error probability lower than $ e_N $, then for any algorithm $ \mathcal{A} $ addressing Problem \ref{problem:PQSI_restatement}, there exists a specific set of quantum states $S = \{\rho_1, \ldots, \rho_K\}$ such that the error probability of algorithm $ \mathcal{A} $ is not lower than $ e_N $.
\end{lemma}

\begin{proof-sketch}
    Suppose that there exists an algorithm $\mathcal{A}$ satisfying whose error probability for solving Problem \ref{problem:PQSI_restatement} is less than $e_N$. We can prove that the algorithm $\mathcal{A}$ can solve the Problem \ref{problem:PRQSI} with the error probability less than $e_N$. Furthermore, we can establish the proof by considering the contrapositive of this statement. Detailed explanations of the proof are included in Appendix \ref{subsection: proof of lemma problem transformation}.
\end{proof-sketch}

According to Lemma \ref{lem:problem transformation}, we will establish the lower bound of the error probability for the problem PQSI by demonstrating the error probability lower bound for the following problem:

\begin{problem}\label{problem:PQDSI_special}
    Consider the Problem \ref{problem:PRQSI}. For $k \in [K]$, let $\alpha_k  = \sqrt{\frac{dz_k -1}{d-1}}$,
    \begin{equation}
        \tau_k(U) = \alpha_k U| 0 \rangle \langle 0|U^\dagger +  \frac{1-\alpha_k}{d} I_d, 
    \end{equation}
    where $U \sim \Haar$.
\end{problem}

Then, In the Problem \ref{problem:PQDSI_special}, for $k \in \{1,...,K\}$ and $U \sim \Haar$, the purity of the quantum state $\tau_k(U)$ satisfies:
\begin{equation*}
    \Tr\left((\tau_k(U))^2\right) = \left(\frac{1 + (d-1)\alpha_k}{d}\right)^2 + (d-1) \left(\frac{1-\alpha_k}{d}\right)^2 = z_k.
\end{equation*}

\subsection{Random quantum state purity certification}

To analyze Problem \ref{problem:PQDSI_special}, we first study the properties of the measurement results obtained from conducting $ N' $ measurements on the sampled quantum states from the quantum state distribution $\mathcal{D}=\{\rho| \rho = \alpha U| 0 \rangle \langle 0|U^\dagger +  \frac{1-\alpha}{d} I_d \}$, using a specific POVM $\mathcal{M} =\{M_0, M_1\}$,
%\begin{equation}
%    \label{eq:Haar state}
%    \mathcal{D}:\ \rho = \alpha U| 0 \rangle \langle 0|U^\dagger +  \frac{1-\alpha}{d} I_d,
%\end{equation}
where $\ U \sim \Haar$ and $\alpha$ is a constant satisfying $0\leq \alpha \leq 1$.

\begin{lemma}\label{lem:larger_expectation}
    Let $a \in [0,1]$. Using a specific POVM $\mathcal{M} =\{M_0, M_1\}$ to measure the random quantum state in $\mathcal{D}=\{\rho| \rho = \alpha U| 0 \rangle \langle 0|U^\dagger +  \frac{1-\alpha}{d} I_d \}$. Let $M = \arg\min_{M' \in\{M_0,M_1\}} \Tr(M')$. We have
    \begin{equation*}
        \mathbb{P}_{U \sim \Haar} \left[ \left|p_{\mathcal{M}}(M| U) - \frac{\Tr(M)}{d}\right| < \frac{2a\sqrt{\Tr(M)}}{d}\right] \geq \frac{3}{4},
    \end{equation*}
    and there is a function $c(\mathcal{M},U)$ satisfying
    %\begin{equation*}
        $p_{\mathcal{M}}(M| U)- \frac{\Tr(M)}{d} =  c(\mathcal{M},U)\alpha$.
    %\end{equation*}
\end{lemma}
\begin{proof-sketch}
    By utilizing the properties of the Haar unitary matrix, we can calculate the variance of $p_\mathcal{M}(M|U)$ and prove the probabilistic bounds in the lemma using Chebyshev's inequality, thus completing the proof. The proof details are provided in Appendix \ref{subsection: proof of lemma larger expectation}.
\end{proof-sketch}

According to Lemma \ref{lem:larger_expectation}, for a specific POVM $\mathcal{M}$ and unitary matrix $U$, let $\mathbb{U}_{\mathcal{M}}$ denote the set of unitary matrix satisfying that 
\begin{equation*}
    \mathbb{U}_{\mathcal{M}} = \left\{U : \left|p_{\mathcal{M}}(M| U) - \frac{\Tr(M)}{d}\right| < \frac{2\alpha\sqrt{\Tr(M)}}{d}\right\}.
\end{equation*}
We have
%\begin{equation}
    $\mathbb{P}_{U \sim \Haar}(U \in \mathbb{U}_\mathcal{M}) \geq \frac{3}{4}$.
%\end{equation}

The following analysis will focus on the unitary matrices in the set $\mathbb{U}_\mathcal{M}$.
% Then expected probability $M$ ``accepts" the

% \begin{problem}\label{problem:purity_certification}
%     Consider the task of distinguishing between the following two alternatives with a fixed two-outcome single-copy POVM $\mathcal{M} = \{M_0, M_1 = I_d -M_0\}$:
%     \begin{equation}
%         \begin{aligned}
%             H_0:& \rho = \alpha U|0\rangle\langle 0| U^\dagger + \frac{1-\alpha}{d} I_d, \\
%             H_1:& \rho = (\alpha+\epsilon) U|0\rangle\langle 0| U^\dagger + \frac{1-\alpha-\epsilon}{d} I_d,
%         \end{aligned}
%     \end{equation}
%     where $U \sim \Haar$.
% \end{problem} 

\subsection{Error probability lower bound}

In this subsection, we will prove the lower bound of error probability for using algorithms to solve Problem \ref{problem:PQSI_restatement} and Problem \ref{problem:PQDSI_special}. 

Let  $M = \arg\min_{M' \in \{M_0,M_1\}}\Tr(M')$, then we have $\Tr(M) \in [0,d/2]$. In the following theorem, let $\Tr(M) > 16$ in order to make $\Tr(M) - 2\sqrt{\Tr(M)} \geq \frac{1}{2}
\Tr(M) $. % In the following theorem, we establish the error probability lower bound of the purest quantum state identification.

\begin{theorem}
    \label{thm:PQSI lower bound}
     Let $M = \arg\min_{M' \in \{M_0,M_1\}}\Tr(M')$ and $\Tr(M) > 16$. For any algorithm $\mathcal{A}$ to solve the purest quantum state identification using fixed 2-outcome randomly incoherent POVM, there exists a set of quantum states which makes the error probability of $\mathcal{A}$ satisfies
     \begin{equation}
         e_N \geq \exp\left( - O\left( \frac{NH_1}{d}\right)\right),
     \end{equation}
     where $H_1 = \min_{i \in \{2,...,K\}} \frac{\Delta_{(i)}}{i}$.
\end{theorem}
\begin{proof-sketch}
    Let $p^{\mathcal{A}}_{e}(\mathcal{M},U)$ denote the error probability for algorithm $\mathcal{A}$ to solve the problem \ref{problem:PQDSI_special}, with specific unitary matrix $U$ and POVM $\mathcal{M}$. The error probability of $\mathcal{A}$ to solve the problem \ref{problem:PQDSI_special} satisfying
    \begin{equation*}
        \begin{aligned}
            e^\mathcal{A}_N = & \int_{\mathcal{M} \in \mathcal{D}_{\mathcal{M}}} \int_{U \sim \Haar} p^\mathcal{A}_e(\mathcal{M},U) d\mathcal{M} dU \\
            \geq & \int_{\mathcal{M} \in \mathcal{D}_{\mathcal{M}}} \int_{U \sim \Haar} p^\mathcal{A}_e(\mathcal{M},U)\indicator\{U \in \mathbb{U}_\mathcal{M}\} d\mathcal{M} dU.
        \end{aligned}
    \end{equation*}
   
    When the $i$-th quantum state is measured using $\mathcal{M}$, the measurement result follows a Bernoulli distribution with parameter $\Tr(M\rho_U)$. From Lemma \ref{lem:larger_expectation} and the definition of $c(\mathcal{M}, U)$, we have 
    \begin{equation*}
        \Tr(M\rho_i|U) = c(\mathcal{M},U)\alpha_i + \frac{\Tr(M)}{d}.    
    \end{equation*}
    
    If $c(\mathcal{M}, U) >0$, we need to find the Bernoulli distribution with the largest parameter where the parameter of the $i$-th Bernoulli distribution is $\Tr(M\rho_i|U)$, and we have 
    \begin{equation}
        \Tr(M\rho_i|U) - \Tr(M\rho_j|U) = c(\mathcal{M},U) \left[\sqrt {\frac{dz_i -1}{d-1}} - \sqrt{\frac{dz_j -1}{d-1}}\right].
    \end{equation}
    
    Since $\Tr(M) > 16$ and according to the definition of $\mathbb{U}(\mathcal{M})$ and $M$, for $U \in \mathbb{U}$ we have
    \begin{equation}
         \Tr(M\rho_i|U) \in \left[\frac{\Tr(M)}{2d}, 1- \frac{\Tr(M)}{2d}\right],
    \end{equation}
    and 
    \begin{equation}
        1- \frac{\Tr(M)}{2d} \geq \frac{1}{2}.
    \end{equation}
    Then, according to the problem of best arm identification problem with Bernoulli distribution, we can demonstrate that $p_e^\mathcal{A}(U,\mathcal{M}) \geq \exp\left(- O\left( \frac{NH_1}{d}\right)\right)$.
    Then we have
    \begin{equation*}
        \begin{aligned}
            e^\mathcal{A}_N \geq & \int_{\mathcal{M} \in \mathcal{D}_{\mathcal{M}}} \int_{U \sim \Haar} p^\mathcal{A}_e(\mathcal{M},U)\indicator\{U \in \mathbb{U}_\mathcal{M}\} d\mathcal{M} dU \\
            \geq & \exp\left(-O\left( \frac{NH_1}{d}\right)\right)\int_{\mathcal{M} \in \mathcal{D}_{\mathcal{M}}}  \frac{3}{4} d\mathcal{M} \\
            \geq & \exp\left(- O\left( \frac{NH_1}{d}\right)\right).
        \end{aligned}
    \end{equation*}
    
    For any algorithm $\mathcal{A}_{\mathcal{D}}$ addressing Problem \ref{problem:PQDSI_special} cannot identify the purest random quantum state with an error probability lower than $\exp\left(-O\left( \frac{NH_1}{d}\right)\right)$. According to Lemma \ref{lem:problem transformation}, we can complete the proof. The proof details are provided in Appendix \ref{subsec: thm lower bound}.
\end{proof-sketch}
\section{PQSI with coherent measurement}
\label{sec: PQSI_co}

In this section, we investigate the problem of purest quantum state identification with coherent measurement and propose an algorithm to solve the purest quantum state identification with coherent measurement based on the SWAP test.

The SWAP test is a quantum algorithm designed to assess the similarity between two quantum states. It offers a method for estimating these states' fidelity to quantify their closeness. We use the SWAP test in Figure \ref{fig:swap_test} to estimate the purity of the quantum state $\rho$ in the unknown quantum state set $S$. The measurement results in Figure \ref{fig:swap_test} have a probability of $\frac{1+ \Tr(\rho^2)}{2}$ for 0 and a probability of $\frac{1- \Tr(\rho^2)}{2}$ for 1. The details of the algorithm are shown in Algorithm \ref{alg:PQSI_coherent}.% , and we have the following results:

\begin{figure}[htb]
    \vskip 0in
    \centering
    \[
    \Qcircuit @C=.7em @R=1.8em {
    \lstick{\rho}      & \qw & \qw        & \multigate{1}{\mathrm{SWAP}} & \qw       &  \qw    \\
    \lstick{\rho}      & \qw & \qw        & \ghost{\mathrm{SWAP}}        & \qw       &  \qw     \\
    \lstick{|0\rangle} & \qw & \gate{H}   & \ctrl{-1}            & \gate{H}  & \meter{} \\
    }
    \]
    \caption{The SWAP test circuit.}
    \label{fig:swap_test}
    % \vskip -0.
\end{figure}
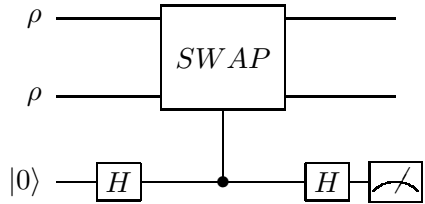

\begin{algorithm}[htb]
    \caption{Coherent measurement based algorithm for solving PQSI problem (CM-PQSI)}
    \label{alg:PQSI_coherent}
    \begin{algorithmic}
    \STATE {\bfseries Input:} Copy access to $S = \{\rho_1,...,\rho_K\}$, sample number $N$.
    \STATE {\bfseries Initialization:} Set $S_0 = \{\rho_1,...,\rho_K\}$, $\overline{\log}(K) = \frac{1}{2} + \sum_{i=2}^K \frac{1}{i}$, $N_0=0$ and $N_k = \left\lceil \frac{1}{2\overline{\log}(K)} \frac{N-K}{K+1-k}\right\rceil$, for $k \in \{1,...,K-1\}$. 
    \FOR{$i=1,..., K-1$}
        \STATE For all $\sigma \in S_{i-1}$, use SWAP test as Figure \ref{fig:swap_test} for $N_k - N_{k-1}$ rounds, and set the outputs as $x_{(\sigma, N_{k-1}+1)},..., x_{(\sigma, N_{k})}$.
        \STATE For all $\sigma \in S_{i-1}$, let $w(\sigma, k) = \frac{1}{N_k}\sum_{i=1}^{N_k} x_{(\sigma, i)}$.
        \STATE Let $S_k = S_{k-1} \setminus \arg \min_{\rho \in S_{k-1}} w(\rho,k). $
    \ENDFOR

    \STATE Output the quantum state $\rho'$ in $S_{k-1}$.

\end{algorithmic}
\end{algorithm}

\begin{theorem} \label{thm:alg_CMPQSI_coherent}
    The probability of error of Algorithm \ref{alg:PQSI_coherent} satisfies
    \begin{equation}
         e_N \leq \frac{K(K-1)}{2}\exp\left(- \frac{N H_2}{8\overline{\log}(K) } \right),
    \end{equation}
    where $H_2 = \min_{i \in \{2,...,K\}} \frac{\Delta^2_{(i)}}{i}$.
\end{theorem}
\begin{proof-sketch}
    The outputs of the SWAP test are within the range $[0, 1]$ and are independent. Thus, we can apply the Hoeffding inequality to complete the proof. Detailed explanations of the proof are included in Appendix \ref{sec:proof_alg_CMPQSI_coherent}.
\end{proof-sketch}

By comparing the conclusions of Theorem \ref{thm:alg_IMPQSI}, Lemma 2\ref{lem2:PQSI_alg}, and Theorem \ref{thm:alg_CMPQSI_coherent}, we can find that under a mild condition that the purity gap is not quite small, i.e., $\Delta_i \gg \frac{1}{d^2}$, the probability of finding the purest quantum state using coherent measurement is much higher than that using incoherent measurement. This result also reflects the importance of the quantum memory.

\section{Conclusion and Outlook}
\label{sec: conclusion}

In this study, we propose a pivotal problem in quantum testing, termed purest quantum state identification (PQSI). This framework applies to various quantum computing and quantum communication tasks. We develop two distinct algorithms to address this problem under different settings. When the learner utilizes incoherent (single-copy) measurement, the upper bound on the error probability of our algorithm is given by $ \exp\left(- \Omega\left(\frac{N H_1}{\log(K) 2^n }\right) \right) $. When the learner is allowed to use coherent (two-copy) measurement, the upper bound on the error probability is given by $ \exp\left(- \Omega\left(\frac{N H_2}{\log(K) }\right) \right) $. By examining the error probabilities of these two algorithms, we can discern the advantage of the coherent measurement over the incoherent one. Furthermore, we establish that for any algorithm utilizing a randomly fixed incoherent two-outcome POVM to solve the PQSI, its error probability is lower bounded by $ \exp\left( - O\left(\frac{N H_1}{2^n}\right)\right) $. Our results lay the groundwork for further investigations into the best quantum state identification. We aim to establish a lower bound for PQSI problems across all POVM bases in future work. Several open questions remain to be addressed, including identifying the nearest quantum state with minimal trace distance and achieving the best quantum state identification with fixed confidence.

\bibliographystyle{alpha}
\bibliography{main}

\appendix

\section{Auxiliary tools}

\subsection{Probability inequalities for sums of bounded random variables}

In this paper, we utilize the following inequalities, which are provided for the sake of completeness.

\begin{theorem}[Chebyshev's Inequality]
    Let $X$ be any random variable with expected value $\mu =\mathbb{E}[X]$ and finite variance $\Var(X)$. Then, for any real number $\varepsilon >0$:
    \begin{equation}
        \mathbb{P}(|X -\mu| \geq \varepsilon) \leq \frac{\Var(X)}{\varepsilon^2}.
    \end{equation}
\end{theorem}

\begin{theorem}[Hoeffding's Inequality]
     If $X_1, X_2, ...,X_n$ are independent with $\mathbb{P}(a \leq X_i \leq b) = 1$ and common mean $\mu$ then for any $\varepsilon > 0$ 
    \begin{equation}
        \mathbb{P}\left( \left|\frac{1}{n} \sum_{i=1}^n X_i - \mu\right| > \varepsilon\right) \leq 2 \exp\left(\frac{-2n\varepsilon^2}{(b-a)^2}\right).
    \end{equation}
\end{theorem}

\begin{theorem}[Bernstein's Inequality]
    If $X_1,...,X_n$ are independent bounded random variables such that $\mathbb{E}[X_i] = 0$ for all $i \in \{1,...,n\}$ and $\mathbb{P}(|X_i| \leq c) =1$ then, for any $\epsilon>0$,
    \begin{equation}
        \mathbb{P}\left(\left|\frac{1}{n}\sum_{i=1}^n \right| \geq \varepsilon\right) \leq \exp\left(-\frac{n\varepsilon^2}{2 \sigma^2 + 2c\epsilon /3}\right),
    \end{equation}
    where $\sigma^2 = \frac{1}{n}\sum_{i=1}^n \Var(X_i)$.
\end{theorem}

\subsection{Properties of Haar unitary matrix}

For a locally compact topological group, its Haar measure is a unique nonzero left-invariant measure  (or right-invariant, depending on the formulation) under group operations. The Haar unitary matrix is the Haar measure on the unitary matrix group and is the concept of drawing unitary matrices uniformly at random. The formal definition of Haar unitary matrix is as follows:

\begin{definition}
    The Haar unitary matrix is the unique probability measure $\mu_{H}$ that is both left and right invariant over the unitary matrix group, i.e., for all integrable functions $f$ and for all unitary matrix $V$, we have:
    \begin{equation}
        \int_{U \sim \Haar}f(U)dU = \int_{U \sim \Haar}f(UV)dU = \int_{U \sim \Haar}f(VU)dU.
    \end{equation}
\end{definition}

For any unit column vector $\bx \in \mathbb{C}^d$, we have 
\begin{equation}
    \mathbb{E}_{U \sim \Haar}\left[f(U\bx)\right] = \mathbb{E}_{\psi \sim \mathbb{C}^d}\left[f(|\psi\rangle)\right].
\end{equation}

We will use the following lemma to complete our proofs in this paper.

\begin{lemma}[see Lemma 22 of Ref.\cite{anshu2022distributed}]
    \label{lem:Haar_meausrement_expectation}
    Let $A,B,C$ be Hermitian matrices. Then
    \begin{equation}
        \mathbb{E}_{\psi \sim \mathbb{C}^d} \langle \psi| A | \psi \rangle = \frac{1}{d}\Tr(A)
    \end{equation}
    and
    \begin{equation}
        \mathbb{E}_{\psi \sim \mathbb{C}^d} \langle \psi| A | \psi \rangle \langle \psi| B | \psi \rangle = \frac{1}{d(d+1)}(\Tr(A)\Tr(B)+\Tr(AB))
    \end{equation}
    and
    \begin{equation}
        \begin{aligned}
            \mathbb{E}_{\psi \sim \mathbb{C}^d} \langle \psi| A | \psi \rangle \langle \psi| B | \psi \rangle \langle \psi| C | \psi \rangle & = \frac{1}{d(d+1)(d+2)}(\Tr(A)\Tr(B)\Tr(C)+\Tr(AB)\Tr(C) \\
            &+ \Tr(A)\Tr(BC) + \Tr(CA)\Tr(B) + \Tr(ABC)).
        \end{aligned}
    \end{equation}
\end{lemma}

\section{Proof of IMPQSI with incoherent measurement}

\subsection{Proof of property for purity collision}

    \begin{lemma} \label{lem:purity_collision_property}
        The expectation and variance of the purity estimation satisfying
        \begin{equation}
            \mathbb{E} [\tilde{g}]  = \frac{m-1}{m}\sum_{i=0}^{d-1} p_i^2,    
        \end{equation}
        and
        \begin{equation}
            \Var[\tilde{g}] \leq \frac{\mathbb{E} [2\tilde{g}]}{m^2} + \frac{4}{m} \sum_{i=0}^{d-1} p_i^3.
        \end{equation}
    \end{lemma}
    \begin{proof}
        
    \label{proof: purity_collision_property}
    The expectation of $\tilde{g}$ satisfying
    \begin{equation}
        \begin{aligned}
            \mathbb{E} [\tilde{g}] = & \mathbb{E}\left[\frac{1}{m^2} \sum_{i=0}^{d-1} \left[ \sum_{j=1}^m \indicator\{x_j =i\}\right]^2\right] - \frac{1}{m}\\
            = & \mathbb{E}\left[\frac{1}{m^2} \sum_{i=0}^{d-1} \left[ \sum_{j=1}^m \sum_{k=1}^m \indicator\{x_j =i\}\indicator\{x_k =i\}\right]\right] - \frac{1}{m}\\
            = & \mathbb{E}\left[\frac{1}{m^2} \sum_{i=0}^{d-1} \left[ \sum_{j=1}^m \sum_{k=1}^m \indicator\{x_j =i\}\indicator\{x_k =i\}\right]\right] - \frac{1}{m}\\
            = & \mathbb{E}\left[\frac{1}{m^2} \sum_{i=0}^{d-1} \left[ \sum_{j=1}^m \indicator\{x_j =i\}+ \sum_{j=1}^m \sum_{k\neq j} \indicator\{x_j =i\}\indicator\{x_k =i\}\right]\right] - \frac{1}{m}\\
            = & \frac{1}{m^2} \sum_{i=0}^{d-1}  \sum_{j=1}^m \mathbb{E}\left[\indicator\{x_j =i\}\right]+ \frac{1}{m^2} \sum_{j=1}^m \sum_{k\neq j} \mathbb{E} [\indicator\{x_j =i\}\indicator\{x_k =i\}] - \frac{1}{m}\\
            = & \frac{m}{m^2} + \frac{m-1}{m} \sum_{j=1}^m p_j^2 - \frac{1}{m}
            = \frac{m-1}{m} \sum_{j=1}^m p_j^2.
        \end{aligned}
    \end{equation}
    The expectation of $\tilde{g}^2$ satisfying
    \begin{equation*}
        \begin{aligned}
            \mathbb{E} [\tilde{g}^2] = & \mathbb{E}\left[\left[\frac{1}{m^2}  \sum_{i=0}^{d-1}\left[ \sum_{j=1}^m \indicator\{x_j =i\}\right]^2 -\frac{1}{m}\right]^2\right] \\
            = & \frac{1}{m^4} \mathbb{E} \left[\sum_{j_1\neq j_2,l_1 \neq l_2}^m\left[\sum_{i,k}^{d-1} \indicator\{x_{j_1} =i\} \indicator\{x_{j_2} =i\} \indicator\{x_{l_1} =k\} \indicator\{x_{l_2} =k\}\right] \right]\\
            = & \frac{1}{m^4} \left[ 4m(m-1)\mathbb{E}[\tilde{g}] + m(m-1)(m-2)(m-3)\mathbb{E}[\tilde{g}]^2 + 4m(m-1)(m-2)\sum_{j=1}^n p_j^3\right]\\
            \leq & \frac{2}{m^2} \mathbb{E}[\tilde{g}] + \mathbb{E}[\tilde{g}]^2+ \frac{4}{m} \sum_{i=0}^{d-1} p_i^3
        \end{aligned}
    \end{equation*}

    Then, we have
    
    \begin{equation}
        \begin{aligned}
            \Var[\tilde{g}] & = \mathbb{E}[\tilde{g}^2] - \mathbb{E}[\tilde{g}]^2\\
            & \leq \frac{2}{m^2} \mathbb{E}[\tilde{g}] + \frac{4}{m} \sum_{i=0}^{d-1} p_i^3.
        \end{aligned}
    \end{equation}

    \end{proof}

\subsection{proof of Theorem \ref{thm:alg_IMPQSI}}
\label{subsec:PQSI_alg1}

    By using the techniques similar to  \cite{anshu2022distributed}, we can prove the following lemma:
    \begin{lemma}[See Lemma 16 of \cite{anshu2022distributed}]
    \label{lem: g_var}
        The expectation of $w(\rho,k)$ and $\tilde{g}_{\rho,j}$ in Algorithm \ref{alg:PQSI_incoherent} satisfies
        \begin{equation}
            \mathbb{E}[w(\rho,k)] = \mathbb{E}[\tilde{g}_{\rho,j}] = \frac{(m-1)(1+ \Tr(\rho^2))}{m(d+1)},
        \end{equation}
        
        and the variance of $\tilde{g}(\rho,j)$ satisfies
        \begin{equation}
            \Var(\tilde{g}(\rho,j)) = O\left(\frac{1}{d^3} + \frac{1}{m^2d} + \frac{1}{md^2}\right).
        \end{equation}
    \end{lemma}

    By the definition of $w(\cdot,\cdot)$ and the definition of $\Delta_{(\cdot)}$, we have
    \begin{equation*}
        \begin{aligned}
            & \mathbb{P}(w({\rho^\star,k}) \leq w({\rho_{(i)},k})) \\
            = & \mathbb{P}\left( (w(\rho_{(i)},k) -  w(\rho^\star,k))  \geq \frac{(m-1)\Delta_{(i)}}{m}\right).
        \end{aligned}
    \end{equation*}
    Since $w(\cdot,\cdot) \in [0,1]$, by Lemma \ref{lem: g_var} and Bernstein's inequality, we have
    \begin{equation}
        \begin{aligned}
            & \mathbb{P}\left( (w(\rho_{(i)},k) -  w(\rho^\star,k))  \geq \frac{(m-1)\Delta_{(i)}}{m}\right) \\
            \leq & \exp\left(-\frac{\lfloor \frac{N_k}{m}\rfloor \left(\frac{m-1}{m(d+1)}\Delta_{(i)}\right)^2}{O(\frac{1}{d^3} + \frac{1}{m^2d} + \frac{1}{md^2}) + \frac{2\Delta_{(i)}}{3(d+1)}}\right) \\
            \leq & \exp\left( -\Omega\left( \frac{\sqrt{c}N_k \Delta_{(i)}}{d}\right)\right).
        \end{aligned}
    \end{equation}

    By a union bound of error probability, we have 

    \begin{equation} \label{eq:err_upper1}
        \begin{aligned} 
            e_n & \leq \sum_{k=1}^{K-1}\sum_{i = K+1-k}^K \mathbb{P}(w({\rho^\star,k}) \leq w({\rho_{(i)},n_k})) \\
            & \leq \sum_{k=1}^{K-1}\sum_{i = K+1-k}^K \exp\left( - \Omega\left(\frac{\sqrt{c}N_k \Delta_{(i)}}{d}\right)\right) \\
            &  \leq \sum_{k=1}^{K-1} k \exp\left( -\Omega\left( \frac{\sqrt{c}N_k \Delta_{(K+1-k)}}{d}\right)\right). \\
        \end{aligned}
    \end{equation}
    By definition of $N_k$, we have
    \begin{equation} \label{eq:err_upper2}
        \begin{aligned}
            & \frac{\sqrt{c}N_k \Delta_{(K+1-k)}}{d} \\
            = & \left\lceil \frac{\sqrt{c}}{\overline{\log}(K)} \frac{N-K}{K+1-k}\right\rceil\frac{\Delta_{(K+1-k)}}{d} \\
            = &\Theta\left( \frac{\sqrt{c}N}{\overline{\log}(K)} \times \frac{\Delta_{(K+1-k)}}{K+1-k}\right). \\
        \end{aligned}
    \end{equation}
    Combining equation \eqref{eq:err_upper1} and \eqref{eq:err_upper2}, we have
    \begin{equation*}
        \begin{aligned}
            e_n & \leq \sum_{k=1}^{K-1} k \exp\left( -\Omega\left(\frac{\sqrt{c}N}{\overline{\log}(K)d} \times\frac{\Delta_{(K+1-k)}}{K+1-k}\right)\right) \\
            & \leq \frac{K(K-1)}{2} \exp\left( -\Omega\left(\frac{\sqrt{c}NH_1}{\overline{\log}(K)d}\right)\right)
        \end{aligned}
    \end{equation*}
    where $H_1 = \min_{i \in \{1,...,K\}} \frac{\Delta_{(i)}}{i}$.

\subsection{proof of Lemma \ref{lem2:PQSI_alg}}
\label{subsec:PQSI_alg2}

\begin{proof}
    By the definition of $w(\cdot,\cdot)$ and the definition of $\Delta_{(\cdot)}$, we have
    \begin{equation*}
        \begin{aligned}
            & \mathbb{P}(w({\rho^\star,k}) \leq w({\rho_{(i)},k})) \\
            = & \mathbb{P}\left( (w(\rho_{(i)},k) -  w(\rho^\star,k))  \geq \frac{(m-1)\Delta_{(i)}}{m}\right).
        \end{aligned}
    \end{equation*}
    Since $w(\cdot,\cdot) \in [0,1]$, by Lemma \ref{lem: g_var} and Bernstein's inequality, we have
    \begin{equation}
        \begin{aligned}
            & \mathbb{P}\left( (w(\rho_{(i)},k) -  w(\rho^\star,k))  \geq \frac{(m-1)\Delta_{(i)}}{m}\right) \\
            \leq & \exp\left(-\frac{\lfloor \frac{N_k}{d}\rfloor \left(\frac{m-1}{m(d+1)}\Delta_{(i)}\right)^2}{O(\frac{1}{d^3}) + \frac{2\Delta_{(i)}}{3(d+1)}}\right) \\
            \leq & \exp\left( -\Omega\left( \min\left( \frac{N_k \Delta_{(i)}}{d^2}, N_k \Delta_{(i)}^2\right)\right)\right).
        \end{aligned}
    \end{equation}

    By a union bound of error probability, we have 

    \begin{equation} \label{eq:error_upper1}
        \begin{aligned} 
            e_n & \leq \sum_{k=1}^{K-1}\sum_{i = K+1-k}^K \mathbb{P}(w({\rho^\star,k}) \leq w({\rho_{(i)},n_k})) \\
            & \leq \sum_{k=1}^{K-1}\sum_{i = K+1-k}^K \exp\left( - \Omega\left( \min\left( \frac{N_k \Delta_{(i)}}{d^2}, N_k \Delta_{(i)}^2\right)\right)\right) \\
            &  \leq \sum_{k=1}^{K-1} k \exp\left( -\Omega\left( \min\left( \frac{N_k \Delta_{(K+1-k)}}{d^2}, N_k \Delta_{(K+1-k)}^2\right)\right)\right). \\
        \end{aligned}
    \end{equation}
    By definition of $N_k$, and combining equation \eqref{eq:error_upper1}, we have
    \begin{equation*}
    \begin{aligned}
        e_N  \leq \frac{K(K-1)}{2} \exp\left(- \Omega\left(\min\left(\frac{N H_2}{\overline{\log}(K)}, \frac{N H_1}{\overline{\log}(K)d^2}\right)\right) \right),
    \end{aligned}
    \end{equation*}
    where $H_1 = \min_{i \in \{2,...,K\}} \frac{\Delta_{(i)}}{i}$, and $H_2 = \min_{i \in \{2,...,K\}} \frac{\Delta^2_{(i)}}{i}$.

\end{proof}

\section{Proof of lower bound}

\subsection{Proof of Lemma \ref{lem:problem transformation}}
\label{subsection: proof of lemma problem transformation}

Suppose that there exists such an algorithm $\mathcal{A}$ satisfying that the error probability of $\mathcal{A}$ for solving Problem \ref{problem:PQSI_restatement} with the quantum state set $S_{\rho}$ whose error probability is less than $e_N$. Let $p_{S_\rho}(x_1,y_1;...;x_N,y_N;z)$ denote the probability of the event satisfying
\begin{enumerate}
    % \item the algorithm $\mathcal{A}$ use the POVM $\mathcal{M}$ to measure quantum states in $S_\rho$;
    \item for $i \in {1,...,N}$, in the round $N$, the algorithm $\mathcal{A}$ select the $x_i$-th quantum state for measurement, and its output is $y_i$;
    \item the algorithm $A$ output $z$-th quantum state at the end. 
\end{enumerate}
Furthermore, let $q_{S_\rho}(x_1,y_1;...;x_N,y_N;z)$ denote the error probability corresponding to $p_{S_\rho}(x_1,y_1;...;x_N,y_N;z)$.

Since when using $\mathcal{A}$ to solve the Problem \ref{problem:PQSI_restatement}, error probability is less than $e_N$. Then for all quantum state set $S =\{\rho_1,...,\rho_K\}$, we have
\begin{equation}
    \int_{(x_1,y_1,...,x_N,y_N,z)} q(x_1,y_1;...;x_N,y_N;z) dp(x_1,y_1;...;x_N,y_N;z) \leq e_N.
\end{equation}

When the learner use the algorithm $\mathcal{A}$ to solve the problem \ref{problem:PRQSI}, its error probability satifying
\begin{equation}
    \begin{aligned}
        e^\mathcal{D}_{N} & = \int_{x \sim \mathcal{D'}} \int_{(x_1,y_1,...,x_N,y_N,z)} q_{S_{\mathcal{D}(x)}}(x_1,y_1;...;x_N,y_N;z) dp_{S_{\mathcal{D}(x)}}(x_1,y_1;...;x_N,y_N;z) dx\\
        & \leq \int_{x \sim \mathcal{D}'} e_N d_x \\
        & \leq e_N.
    \end{aligned}
\end{equation}

then we can prove that if there is an algorithm $\mathcal{A}$ can solve the Problem \ref{problem:PRQSI} with the error probability less than $e_N$, then it can solve the Problem \ref{problem:PQSI_restatement} with the error probability less than $e_N$. Furthermore, we can establish the proof by considering the contrapositive of this statement.

\subsection{proof of Lemma \ref{lem:larger_expectation}}
\label{subsection: proof of lemma larger expectation}
    Without loss of generality, assume that $M_0 = \arg\min_{M' \in \{M_0,M_1\}} \Tr(M')$.  According to the definition of POVM, there exists a unitary matrix \( V \) and a diagonal matrix $ \Sigma_0 = \diag(b_0,...,b_{d-1})$, where $b_0,...,b_{d-1} \in [0,1]$ such that
    \begin{equation*}
        \begin{aligned}
            M_0 & = V \Sigma_0 V^\dagger = \sum_{i=0}^{d-1} b_i V |i \rangle \langle i | V^\dagger, \\
            M_1 & = I - V \Sigma_0 V^\dagger = \sum_{i=0}^{d-1} (1-b_i) V |i \rangle \langle i | V^\dagger. \\
        \end{aligned}
    \end{equation*}
    We have
    \begin{equation}
        \begin{aligned}
            & \Tr(M_0) = \Tr (\Sigma_0) = \sum_{i=0}^{d-1} b_i, \\     
            & \Tr(M_1) = \Tr (I - \Sigma_0) = d - \sum_{i=0}^{d-1} b_i.
        \end{aligned}
    \end{equation}
    Let $ p_{\mathcal{M}}(M | U) $ denote the probability that $ M $ ``accepts" the quantum state $ \alpha U | 0 \rangle \langle 0 | U^\dagger + \frac{1-\alpha}{d-1} I_d $. According to the property of the Haar measure and the identity matrix $I_d$, we have 
    \begin{equation}
        \begin{aligned}
            & \mathbb{E}_{U \sim \Haar}\left[ p^2_{\mathcal{M}}(M_0| U)\right] \\
            = & \mathbb{E}_{U \sim \Haar}\left[\Tr^2\left(M_0 \left( \alpha U | 0 \rangle \langle 0 | U^\dagger + \frac{1-\alpha}{d} I_d\right)\right)\right] \\
            = & \mathbb{E}_{U \sim \Haar}\left[\Tr^2\left(\sum_{i=0}^{d-1} b_i V |i \rangle \langle i | V^\dagger \left( \alpha U | 0 \rangle \langle 0 | U^\dagger + \frac{1-\alpha}{d} I_d\right)\right)\right] \\
            = & \mathbb{E}_{U \sim \Haar}\left[\left(\sum_{i=0}^{d-1} b_i  \langle i | V^\dagger \left( \alpha U | 0 \rangle \langle 0 | U^\dagger + \frac{1-\alpha}{d} I_d)\right)V |i \rangle\right)^2\right] \\
            = & \mathbb{E}_{U \sim \Haar}\left[\left(\sum_{i=0}^{d-1}  \left(\alpha b_i  \langle i | V^\dagger U | 0 \rangle \langle 0 | U^\dagger V |i \rangle\right) + \sum_{i=0}^{d-1} b_i\frac{1-\alpha}{d} \langle i | V^\dagger I_dV |i \rangle\right)^2\right] \\
            = & \mathbb{E}_{U \sim \Haar}\left[\left(\sum_{i=0}^{d-1} \left(\alpha b_i  \langle i |U | 0 \rangle \langle 0 | U^\dagger|i \rangle\right) + \frac{1-\alpha}{d}\Tr(M_0)\right)^2\right] \\
        \end{aligned}
    \end{equation}
    Let $V_i$ is the matrix satisfying that $V_i|i\rangle = |0\rangle$, then $V_i$ is an unitary matrix and $V_i^{-1} = V_i^\dagger$, we have
    \begin{equation} \label{eq2:p_M0}
        \begin{aligned}
            &\mathbb{E}_{U \sim \Haar}\left[ p^2_{\mathcal{M}}(M_0| U)\right] \\
            = & \mathbb{E}_{U \sim \Haar}\left[\left(\sum_{i=0}^{d-1} \left(\alpha b_i  \langle i |U | 0 \rangle \langle 0 | U^\dagger|i \rangle\right) + \frac{1-\alpha}{d}\Tr(M_0)\right)^2\right] \\
            = & \mathbb{E}_{U \sim \Haar}\left[\left(\sum_{i=0}^{d-1} \left(\alpha b_i  \langle 0 |U^\dagger | i \rangle \langle i | U|0 \rangle\right) + \frac{1-\alpha}{d}\Tr(M_0)\right)^2\right] \\
            = & \mathbb{E}_{U \sim \Haar}\left[\left(\sum_{i=0}^{d-1} \left(\alpha b_i  \langle 0 |U | i \rangle \langle i | U^\dagger|0 \rangle\right) + \frac{1-\alpha}{d}\Tr(M_0)\right)^2\right] \\
            = & \mathbb{E}_{\psi \sim \mathbb{C}^d} \left[\left(\sum_{i=0}^{d-1} \left(\alpha b_i  \langle \psi | i \rangle \langle i |\psi \rangle\right) + \frac{1-\alpha}{d}\Tr(M_0)\right)^2\right] \\
            = & \mathbb{E}_{\psi \sim \mathbb{C}^d} \left[  \alpha^2 \sum_{i=0}^{d-1} b^2_i  \langle \psi | i \rangle \langle i |\psi \rangle\langle \psi | i \rangle \langle i |\psi \rangle + \alpha^2\sum_{i=0}^{d-1}\sum_{j \neq i}b_ib_j \langle \psi | i \rangle \langle i |\psi \rangle\langle \psi | j \rangle \langle j |\psi \rangle \right.\\
            &\left. + 2\sum_{i=0}^{d-1} \left(\alpha b_i  \langle \psi | i \rangle \langle i |\psi \rangle\right)\frac{1-\alpha}{d}\Tr(M_0) + \left(\frac{1-\alpha}{d}\right)^2\Tr^2(M_0)    \right] \\
        \end{aligned}
    \end{equation}
    According to the Lemma \ref{lem:Haar_meausrement_expectation}, we have for $i,j \in \{0,...,d-1\}$, $i\neq j$,
    \begin{equation} \label{eq:expectation_i}
        \mathbb{E}_{\psi \sim \mathbb{C}^d} \langle \psi | i \rangle \langle i |\psi \rangle= \frac{1}{d},
    \end{equation}
    and
    \begin{equation} \label{eq:expectation_ii}
        \mathbb{E}_{\psi \sim \mathbb{C}^d} \langle \psi | i \rangle \langle i |\psi \rangle\langle \psi | i \rangle \langle i |\psi \rangle= \frac{2}{d(d+1)},
    \end{equation}
    and similarly
    \begin{equation} \label{eq:expectation_ij}
        \mathbb{E}_{\psi \sim \mathbb{C}^d} \langle \psi | i \rangle \langle i |\psi \rangle\langle \psi | j \rangle \langle j |\psi \rangle= \frac{1}{d(d+1)}.
    \end{equation}
    According to Equation \eqref{eq2:p_M0}, \eqref{eq:expectation_i},\eqref{eq:expectation_ii} and \eqref{eq:expectation_ij}, we have 
    \begin{equation}
        \begin{aligned} \label{eq3:p_M0}
            &\mathbb{E}_{U \sim \Haar}\left[ p^2_{\mathcal{M}}(M_0| U)\right] \\
            = & \frac{2\alpha^2}{d(d+1)} \sum_{i=0}^{d-1} b^2_i   + \frac{\alpha^2}{d(d+1)}\sum_{i=0}^{d-1}\sum_{j \neq i}b_ib_j  + \frac{2\alpha(1-\alpha)}{d^2}\Tr(M_0) \sum_{i=0}^{d-1} b_i + \left(\frac{1-\alpha}{d}\right)^2\Tr^2(M_0)  \\
            = & \frac{\alpha^2}{d(d+1)} \sum_{i=0}^{d-1} b^2_i   + \frac{\alpha^2}{d(d+1)}\left(\sum_{i=0}^{d-1}b_i\right)^2  + \frac{2\alpha(1-\alpha)}{d^2}\Tr(M_0) \sum_{i=0}^{d-1} b_i + \left(\frac{1-\alpha}{d}\right)^2\Tr^2(M_0) \\
            = & \frac{\alpha^2}{d(d+1)} \sum_{i=0}^{d-1} b^2_i   + \frac{\alpha^2}{d(d+1)}\Tr^2(M_0)  + \frac{2\alpha(1-\alpha)}{d^2}\Tr^2(M_0) + \left(\frac{1-\alpha}{d}\right)^2\Tr^2(M_0) \\
             = & \frac{\alpha^2}{d(d+1)} \sum_{i=0}^{d-1} b^2_i   + \left[\frac{1}{d^2} - \frac{\alpha^2}{d^2(d+1)}\right]\Tr^2(M_0), \\
        \end{aligned}
    \end{equation}
    and
    \begin{equation}
        \begin{aligned}
            &\mathbb{E}_{U \sim \Haar}\left[ p_{\mathcal{M}}(M_0| U)\right] \\
            = & \mathbb{E}_{U \sim \Haar}\left[\Tr\left(M_0 \left( \alpha U | 0 \rangle \langle 0 | U^\dagger + \frac{1-\alpha}{d} I_d\right)\right)\right] \\
            = & \mathbb{E}_{U \sim \Haar}\left[\Tr\left(\sum_{i=0}^{d-1} b_i V |i \rangle \langle i | V^\dagger \left( \alpha U | 0 \rangle \langle 0 | U^\dagger + \frac{1-\alpha}{d} I_d\right)\right)\right] \\
            = & \mathbb{E}_{U \sim \Haar}\left[\Tr\left(\sum_{i=0}^{d-1} b_i V |i \rangle \langle i | V^\dagger \left( \alpha U | 0 \rangle \langle 0 | U^\dagger + \frac{1-\alpha}{d} I_d\right)\right)\right]
            = & \frac{\Tr(M_0)}{d}.
        \end{aligned}
    \end{equation}
    % Similarly, we have
    % \begin{equation} \label{eq:p_M1}
    %     \mathbb{E}_{U \sim Haar}\left[ p^2_{\mathcal{M}}(M_1| U)\right] = \frac{a^2}{d(d+1)} \sum_{i=0}^{d-1} (1 -b_i)^2   + \left[\frac{1}{d^2} - \frac{a^2}{d^2(d+1)}\right]\Tr^2(M_1).
    % \end{equation}

Then the variance of $p_{\mathcal{M}}(M_0| U)$ is given by 

    \begin{equation}
        \begin{aligned}
            & \Var\left[p_{\mathcal{M}}(M_0| U)\right] \\
            = & \mathbb{E}_{U \sim \Haar}\left[ p^2_{\mathcal{M}}(M_0| U)\right] - \left(\mathbb{E}_{U \sim \Haar}\left[ p_{\mathcal{M}}(M_0| U)\right]\right)^2 \\
            = &\frac{\alpha^2}{d(d+1)} \sum_{i=0}^{d-1} b^2_i   + \left[\frac{1}{d^2} - \frac{\alpha^2}{d^2(d+1)}\right]\Tr^2(M_0) - \frac{\Tr^2(M_0)}{d^2} \\
            = &\frac{\alpha^2}{d(d+1)} \sum_{i=0}^{d-1} b^2_i   - \frac{\alpha^2}{d^2(d+1)}\Tr^2(M_0) \\
            = &\frac{\alpha^2}{d^2(d+1)} \left[d\sum_{i=0}^{d-1} b^2_i - \Tr^2(M_0)\right] \\
            \leq & \frac{\alpha^2}{d^2(d+1)} \left[d\sum_{i=0}^{d-1} b_i - \Tr^2(M_0)\right] \\
            = & \frac{\alpha^2}{d^2(d+1)} \left[d\Tr(M_0) - \Tr^2(M_0)\right] \\
            \leq & \frac{\alpha^2}{d^2(d+1)} \left[d\Tr(M_0) - \Tr^2(M_0)\right] \\ 
            \leq & \frac{\alpha^2\Tr(M_0)}{d(d+1)}.
        \end{aligned}
    \end{equation}

    From Chebyshev's Inequality, we have
    \begin{equation}
        \mathbb{P}_{U \sim \Haar} \left[ \left|p_{\mathcal{M}}(M_0| U) - \frac{\Tr(M_0)}{d}\right| \geq \frac{2\alpha\sqrt{\Tr(M_0)}}{d}\right]  < \frac{1}{4}.
    \end{equation}

    And
    \begin{equation}
        \begin{aligned}
            & p_{\mathcal{M}}(M_0| U)  \\
            = & \Tr\left(\sum_{i=0}^{d-1} b_i V |i \rangle \langle i | V^\dagger \left( \alpha U | 0 \rangle \langle 0 | U^\dagger + \frac{1-\alpha}{d} I_d\right)\right) \\
            = & \Tr\left(\sum_{i=0}^{d-1} b_i V |i \rangle \langle i | V^\dagger \left( \alpha U | 0 \rangle \langle 0 | U^\dagger + \frac{\alpha}{d} I_d\right)\right) + \Tr\left(\sum_{i=0}^{d-1} b_i V |i \rangle \langle i | V^\dagger \left( \frac{1}{d} I_d\right)\right)\\
            = & \alpha \Tr\left(\sum_{i=0}^{d-1} b_i V |i \rangle \langle i | V^\dagger \left( U | 0 \rangle \langle 0 | U^\dagger + \frac{1}{d} I_d\right)\right) + \frac{M_0}{d}.\\
        \end{aligned}
    \end{equation}
    Let $c(\mathcal{M},U) = \Tr\left(\sum_{i=0}^{d-1} b_i V |i \rangle \langle i | V^\dagger \left( U | 0 \rangle \langle 0 | U^\dagger + \frac{1}{d} I_d\right)\right)$, we have
    \begin{equation}
        p_{\mathcal{M}}(M_0| U) - \frac{M_0}{d} =  c(\mathcal{M},U) \alpha.
    \end{equation}

\subsection{Proof of Theorem \ref{thm:PQSI lower bound}}
\label{subsec: thm lower bound}

    We will use the following theorem to complete the proof:
    
    \begin{theorem}[see Theorem 4 of Ref. \cite{audibert2010best}]
        \label{theorem:classical_BAI}
        Let $\nu_1,...,\nu_K$ be Bernoulli distributions with parameters in $[a,1-a]$, $a \in(0,1/2)$. For any forecaster, there exists a permutation $\sigma: \{1,..., K\} \rightarrow \{1,..., K\}$ such that the probability error of the forecaster on the bandit problem defined by $\tilde{\nu}_1= \nu_{\sigma(1)},..., \tilde{\nu}_K= \nu_{\sigma(K)}$ satisfies
        \begin{equation}
            e_n \geq \exp\left(- \frac{(5+o(1))nH}{a(1-a)}\right),
        \end{equation}
        where $H =\min_i \frac{(\mathbb{E}[\nu^*] - \mathbb{E}[\nu_{(i)}])^2}{i}$.
    \end{theorem}

    Let $p^{\mathcal{A}}_{e}(\mathcal{M},U)$ denote the error probability for algorithm $\mathcal{A}$ to solve the problem \ref{problem:PQDSI_special}, with specific unitary matrix $U$ and POVM $\mathcal{M}$, and  $M = \min_{M' \in \{M_0,M_1\}} \Tr(M')$. Then the error probability of $\mathcal{A}$ to solve the problem \ref{problem:PQDSI_special} satisfying
    \begin{equation}
        \label{eq1:proof_lower_bound}
        \begin{aligned}
            e^\mathcal{A}_N = & \int_{\mathcal{M} \in \mathcal{D}_{\mathcal{M}}} \int_{U \sim \Haar} p^\mathcal{A}_e(\mathcal{M},U) d\mathcal{M} dU \\
            \geq & \int_{\mathcal{M} \in \mathcal{D}_{\mathcal{M}}} \int_{U \sim \Haar} p^\mathcal{A}_e(\mathcal{M},U)\indicator\{U \in \mathbb{U}_\mathcal{M}\} d\mathcal{M} dU.
        \end{aligned}
    \end{equation}
    The first line corresponds to the deifintion of $e_N^\mathcal{A}$, the second line corresponds to that $p_{e}^\mathcal{A}(\mathcal{M},U) \geq 0$ and $\mathbb{U}_{\mathcal{M}}$ is a subset of unitary matrix.
    
    When the $i$-th quantum state is measured using $\mathcal{M}$, the process in which the output result is accepted by $M$ follows a Bernoulli distribution with parameter $\Tr(M\rho_U)$. From Lemma \ref{lem:larger_expectation} and the definition of $c(\mathcal{M}, U)$, we have 
    \begin{equation}
        \Tr(M\rho_i|U) = c(U,\mathcal{M})\alpha_i + \frac{\Tr(M)}{d}.    
    \end{equation}
    If $c(\mathcal{M},U) >0$, we need to find the Bernoulli distribution with the largest parameter where the parameter of the $i$-th Bernoulli distribution is $\Tr(M\rho_i|U) = c(\mathcal{M},U)\alpha_i + \frac{\Tr(M)}{d}$. Then we have
    \begin{equation}
        \begin{aligned}
            & \Tr(M\rho_i|U) - \Tr(M\rho_j|U) \\
            = & c(\mathcal{M},U)\alpha_i  - c(\mathcal{M},U)\alpha_j \\
            = & c(\mathcal{M},U) \left[\sqrt {\frac{dz_i -1}{d-1}} - \sqrt{\frac{dz_j -1}{d-1}}\right].
        \end{aligned}
    \end{equation}
    Since $\Tr(M) > 16$, for $U \in \mathbb{U}(\mathcal{M})$ we have
    \begin{equation}
        \begin{aligned}
            \Tr(M\rho_i|U) & = c(\mathcal{M},U)\alpha_i + \frac{\Tr(M)}{d} \\
            & \geq - \frac{2\sqrt{\Tr(M)}}{d} + \frac{Tr(M)}{d} \geq \frac{\Tr(M)}{2d}.
        \end{aligned}
    \end{equation}
    And since $M = \arg\min_{M'\in \{M_0,M-1\}} \Tr(M')$, we have $\Tr(M) \leq \frac{1}{2}$, then for $U \in \mathbb{U}$ we have
    \begin{equation}
         \Tr(M\rho_i|U) \in \left[\frac{\Tr(M)}{2d}, 1-\frac{\Tr(M)}{2d}\right],
    \end{equation}
    and 
    \begin{equation}
        1- \frac{\Tr(M)}{2d} \geq \frac{1}{2}.
    \end{equation}
    According to Theorem \ref{theorem:classical_BAI} and the definition of $\mathbb{U}_\mathcal{M}$, for $U \in \mathbb{U}_M$ we have
    \begin{equation}\label{eq:proof_lower_bound1}
        \begin{aligned}
             p_e^\mathcal{A}(\mathcal{M},U) \geq & \exp\left(-O\left( \frac{N}{\Omega\left(\frac{\Tr(M)}{d}\right)\left(1-\Omega\left(\frac{\Tr(M)}{d}\right)\right)}\min_i \frac{\left(\Tr(M\rho^\star|U) - \Tr(M\rho_{(i)}|U)\right)^2}{i}\right)\right) \\
            = & \exp\left(-O\left( \frac{N c^2(\mathcal{M},U)}{\Omega\left(\frac{\Tr(M)}{d}\right)}\min_i \frac{\left(\sqrt{dz_{i^\star} -1} - \sqrt{dz_{(i)} -1}\right)^2}{i(d-1)}\right)\right) \\
            = & \exp\left(-O\left( \frac{N c^2(\mathcal{M},U)}{\Omega\left(\frac{\Tr(M)}{d}\right)}\min_i \frac{dz_{i^\star}-1 + dz_{(i)}-1 -2\sqrt{(dz_{i^\star}-1)(dz_{(i)}-1)}}{i(d-1)}\right)\right) \\
            \geq & \exp\left(-O\left( \frac{N c^2(\mathcal{M},U)}{\Omega\left(\frac{\Tr(M)}{d}\right)}\min_i \frac{[dz_{i^\star}-1] - [dz_{(i)}-1]}{i(d-1)}\right)\right) \\
            = & \exp\left(-O\left( \frac{N c^2(\mathcal{M},U)}{\Omega\left(\frac{\Tr(M)}{d}\right)}\min_i \frac{d\Delta_{(i)}}{i(d-1)}\right)\right) \\
        \end{aligned} 
    \end{equation}
    According to the definition of $\mathbb{U}_\mathcal{M}$, for $U \in \mathbb{U}_M$ we have 
    \begin{equation} \label{eq:proof_lower_bound2}
        c(\mathcal{M},U) \in \left(-\frac{2\sqrt{\Tr(M)}}{d}, \frac{2\sqrt{\Tr(M)}}{d}\right).
    \end{equation}
    According to Equation \eqref{eq:proof_lower_bound1} and Equation \eqref{eq:proof_lower_bound2}, we have
    \begin{equation}
    \label{eq2:proof_lower_bound}
        \begin{aligned}
            p_e^\mathcal{A}(U,\mathcal{M}) \geq & \exp\left(-O\left( \frac{N c^2(U,\mathcal{M})}{\Omega\left(\frac{\Tr(M)}{d}\right)}\min_i \frac{d\Delta_{(i)}}{i(d-1)}\right)\right) \\
            \geq & \exp\left(-O\left( \frac{N \left(\frac{\sqrt{\Tr(M)}}{d}\right)^2}{\Omega\left(\frac{\Tr(M)}{d}\right)}\min_i \frac{d\Delta_{(i)}}{i(d-1)}\right)\right) \\
            \geq & \exp\left(-O\left( \frac{N}{d}\min_i \frac{\Delta_{(i)}}{i}\right)\right) \\
            = & \exp\left(-O\left( \frac{NH_1}{d}\right)\right).
        \end{aligned}
    \end{equation}
    Similarly, if $c(U,\mathcal{M}) <=0$, we have
    \begin{equation}
    \label{eq3:proof_lower_bound}
        p_e^\mathcal{A}(U,\mathcal{M}) \geq \exp\left(- O\left(\frac{NH_1}{d}\right)\right).
    \end{equation}
    According to Equation \eqref{eq1:proof_lower_bound}, \eqref{eq2:proof_lower_bound}, \eqref{eq3:proof_lower_bound} we have
    \begin{equation*}
        \begin{aligned}
            e^\mathcal{A}_N \geq & \int_{\mathcal{M} \in \mathcal{D}_{\mathcal{M}}} \int_{U \sim \Haar} p^\mathcal{A}_e(\mathcal{M},U)\indicator\{U \in \mathbb{U}_\mathcal{M}\} d\mathcal{M} dU \\
            \geq & \int_{\mathcal{M} \in \mathcal{D}_{\mathcal{M}}} \int_{U \sim \Haar} \exp\left(-O\left( \frac{NH_1}{d}\right)\right) \indicator\{U \in \mathbb{U}_\mathcal{M}\} d\mathcal{M} dU \\
            \geq & \exp\left(-O\left( \frac{NH_1}{d}\right)\right) \int_{U \sim \Haar}  \indicator\{U \in \mathbb{U}_\mathcal{M}\} d\mathcal{M} dU \\
            \geq & \exp\left(-O\left( \frac{NH_1}{d}\right)\right)\int_{\mathcal{M} \in \mathcal{D}_{\mathcal{M}}}  \frac{3}{4} d\mathcal{M} \\
            \geq & \exp\left(- O\left(\frac{NH_1}{d}\right)\right).\\            
        \end{aligned}
    \end{equation*}
    Then for any algorithm $\mathcal{A}_{\mathcal{D}}$ addressing Problem \ref{problem:PQDSI_special} cannot identify the purest random quantum state with an error probability lower than $\exp\left(-O\left( \frac{NH_1}{d}\right)\right)$. According to Lemma \ref{lem:problem transformation}, we can complete the proof.
\section{Proof of Theorem \ref{thm:alg_CMPQSI_coherent}}
\label{sec:proof_alg_CMPQSI_coherent}

    By the definition of $w(\cdot,\cdot)$ and the definition of $\Delta_{(\cdot)}$, we have
    \begin{equation*}
        \begin{aligned}
            & \mathbb{P}(w({\rho^\star,k}) \leq w({\rho_{(i)},k})) \\
            = & \mathbb{P}\left( (w(\rho_{(i)},k) -  w(\rho^\star,k))  \geq \frac{\Delta_{(i)}}{2}\right).
        \end{aligned}
    \end{equation*}
    Since $x_{(\cdot,\cdot)} \in [0,1]$ and Hoeffding's inequality, we have
    \begin{equation}
        \begin{aligned}
            & \mathbb{P}\left( (w(\rho_{(i)},k) -  w(\rho^\star,k))  \geq \frac{\Delta_{(i)}}{2}\right) \\
            \leq & \exp\left(- \frac{N_k}{2} \left(\frac{\Delta_{(i)}}{2}\right)^2 \right) \\
            \leq & \exp\left( - \frac{N_k \Delta^2_{(i)}}{8}\right).
        \end{aligned}
    \end{equation}

    By a union bound of error probability, we have 

    \begin{equation} \label{eq2:err_upper1}
        \begin{aligned} 
            e_n & \leq \sum_{k=1}^{K-1}\sum_{i = K+1-k}^K \mathbb{P}(w({\rho^\star,k}) \leq w({\rho_{(i)},n_k})) \\
            & \leq \sum_{k=1}^{K-1}\sum_{i = K+1-k}^K \exp\left( - \frac{N_k \Delta^2_{(K+1-k)}}{8}\right) \\
            &  \leq \sum_{k=1}^{K-1} k \exp\left( - \frac{N_k \Delta^2_{(K+1-k)}}{8}\right). \\
        \end{aligned}
    \end{equation}
    By definition of $N_k$, we have
    \begin{equation} \label{eq2:err_upper2}
        \begin{aligned}
             N_k \Delta^2_{(K+1-k)} 
            =  \left\lceil \frac{1}{\overline{\log}(K)} \frac{N-K}{K+1-k}\right\rceil\Delta^2_{(K+1-k)} 
            \leq  \frac{N-K}{\overline{\log}(K)} \times \frac{\Delta^2_{(K+1-k)}}{K+1-k}. 
        \end{aligned}
    \end{equation}
    Combining equation \eqref{eq2:err_upper1} and \eqref{eq2:err_upper2}, we have
    \begin{equation*}
        \begin{aligned}
            e_n & \leq \sum_{k=1}^{K-1} k \exp\left( - \frac{N-K}{8\overline{\log}(K)} \times \frac{\Delta^2_{K+1-k}}{K+1-k}\right) \\
            & \leq \frac{K(K-1)}{2} \exp\left(-\frac{NH_2}{8\overline{\log}(K)}\right),
        \end{aligned}
    \end{equation*}
    where $H_2 = \min_{i \in \{1,...,K\}} \frac{\Delta^2_{(i)}}{i}$.

\end{document}